\documentclass[12pt]{article}

\usepackage{amsmath,amsfonts,amssymb,amsthm}

\title{Non-existence results for vectorial bent functions with Dillon exponent}
\date{}
\author{Lucien Lapierre\\
Petr Lison\v{e}k
\\
Department of Mathematics\\
Simon Fraser University\\
Burnaby, BC, V5A 1S6\\
Canada\\
\  \\
{\tt plisonek@sfu.ca}
}

\newtheorem{theorem}{Theorem}[section]

\newtheorem{proposition}[theorem]{Proposition}

\theoremstyle{definition}
\newtheorem{definition}[theorem]{Definition}
\newtheorem{example}[theorem]{Example}

\def\F{{\mathbb F}}

\def\K{{\mathcal K}}

\def\Tr{{\rm Tr}}

\def\T{{\rm Tr}}
\def\S{{\rm S}}

\def\al{\alpha}

\begin{document}

\maketitle

\begin{abstract}
We prove new non-existence results
for vectorial monomial Dillon type bent functions
mapping the field of order $2^{2m}$
to the field of order $2^{m/3}$. When $m$ is odd
and $m>3$ we show that there are no such functions.
When $m$ is even we derive a condition for the bent coefficient.
The latter result allows us to find examples
of bent functions with $m=6$ in a simple way.
\end{abstract}

\section{Background}

Bent functions are maximally nonlinear functions
from $\F_2^n$ to $\F_2^k$
that have important applications in cryptography,
coding theory, communication sequences
and combinatorial design theory 
\cite{CarletMesnager,GolombGongBook,Sihem}.
%They are also known as perfect nonlinear functions.
For $k=1$,
in the English literature bent functions 
were introduced by Rothaus 
first in the 1960s
in a technical report with restricted circulation,
and then in 1976 in his journal paper.
In 1991 Nyberg introduced bent functions for arbitrary $k$.
Binary bent functions exist only for even $n$ and for $k\le n/2$.
A stronger concept was proposed by Gong and Golomb \cite{GongGolombDES}
and then formally introduced and studied
by Youssef and Gong \cite{youssefgong}:
A function is hyperbent if it is bent after each bijective
monomial substitution on its domain.
We refer the reader 
to the recent survey of bent functions research~\cite{CarletMesnager} 
for an extensive account of this subject.

Let $\F_{2^s}$ denote the field of order $2^s$,
and let $\F_{2^s}^*$ denote the multiplicative group
of its non-zero elements.
For positive integers $s,t$ such that $t|s$
let $\Tr^s_t$ denote the trace function from $\F_{2^s}$ to $\F_{2^t}$,
that is, $\Tr^s_t(x)=\sum_{k=0}^{s/t-1}x^{2^{kt}}$.
When $t=1$ the trace function is called {\em absolute trace.}

Let $f$ be a function from 
$\F_{2^n}$ to $\F_{2^k}$.
For $a \in \F_{2^n}$ and $b \in \F_{2^k}^*$
the {\em  Walsh transform} of $f$ at $(a,b)$ is
\[W_f (a,b) = \sum_{x \in \F_{2^n}} (-1)^{\Tr^k_1(bf(x)) + \Tr^n_1(ax)}.\]

\begin{definition}
Let $n$ be even and let $f:\F_{2^n}\rightarrow\F_{2^k}$.
We say that $f$ is {\em bent} if 
$|W_f (a,b)| =2^{n/2}$ for all $a \in \F_{2^n}$ and all $b \in \F_{2^k}^*$.
\end{definition}

Let $a \in \F_{2^m}$. The {\em Kloosterman sum over $\F_{2^m}$} is defined as 
$$\K_{2^m}(a) := 1+\sum_{x \in \F_{2^m}^*} (-1)^{\Tr^m_1(\frac{1}{x}+ax)}.$$
If $a \in \F_{2^m}^*$ is such that $\K_{2^m}(a) = 0$, then $a$ is called 
a {\em Kloosterman zero in $\F_{2^m}$}.
%It is easy to see that $\K_{2^m}(a)=\K_{2^m}(a^2)$ for all $a\in\F_{2^m}$.

Lison\v{e}k and Moisio 
\cite{KZerosNotInSbflds}
proved that proper subfields of  $\F_{2^m}$ do not
contain Kloosterman zeros, with one exception.

\begin{theorem}
\cite{KZerosNotInSbflds}
\label{thm:KZerosNotInSbflds}
Let $a \in \F_{2^k}^*$. If $\K_{2^{kt}}(a) = 0$ for integer $t>1$, then $kt=4$ and $a = 1$.
\end{theorem}

In his dissertation Dillon proved 
the following remarkable characterization
of a class of monomial bent functions in the trace form.
It was proved later that these functions are also hyperbent
and they have other attractive properties,
such as maximum possible algebraic degree.
\begin{theorem}
\cite{dillon}
\label{cor:dilloncharacterization}
For $a \in \F_{2^m}^*$ the function $f(x) = \Tr^{2m}_1(ax^{2^m-1})$ is 
hyperbent if and only if $\K_{2^m}(a) = 0$.
\end{theorem}

It is natural to extend the investigation of Dillon type functions
to vectorial functions, and study the bentness of such functions.
For $a$ belonging to some extension $L$ of $\F_{2^k}$ 
let $a\F_{2^k}^*:=\{au\, : \, u \in \F_{2^k}^*\}$ 
denote the coset of $\F_{2^k}^*$ in $L^*$ containing~$a$.
The following result is easy.

\begin{proposition}
\label{prop-K-vanishes-on-coset}
\cite[Proposition~5]{LL-ISIT16}
Let $a \in \F_{2^m}^*$ and suppose that $k$ divides $m$. 
Then $f(x) = \Tr^{2m}_k(a x^{2^m-1})$ is bent if and only if 
$\K_{2^m}(u)=0$ for all $u\in a\F_{2^k}^*$.
\end{proposition}

Note that while in general we could take the coefficient $a$ 
to lie in $\F_{2^{2m}}$, the assumption 
$a \in \F_{2^m}^*$ can be made without loss of generality.
This is well known, see for example \cite{LL-ISIT16} for details.

Muratovi\'{c}-Ribi\'{c}, Pasalic and Bajri\'{c}
\cite{Muratovic-Ribic}
proved the following non-existence result.

\begin{theorem}
\cite[Theorem~10]{Muratovic-Ribic}
\label{thm-mrd2}
For $a \in \F_{2^m}^*$ the function
$f(x) = \Tr^{2m}_m(ax^{2^m-1})$ is not bent.
\end{theorem}

We extended this result as follows.

\begin{theorem}
\label{thm:4mtomnonexist}
\cite[Theorem 12]{LL-ISIT16}
Let $m$ be even and let $a \in \F_{2^m}^*$. The function 
$f(x) = \Tr^{2m}_{m/2}(ax^{2^m-1})$ is not bent.
\end{theorem}

\section{New results}

In a natural continuation of the line of research
reported above, 
our new results proved in this paper concern 
Dillon type monomial bent functions from $\F_{2^{2m}}$
to $\F_{2^{m/3}}$ for $m$ divisible by~$3$.

Let $a\in \F_{2^t}$.
From now on let $\T_{2^t}(a)$ denote the absolute trace of $a$,
\[
\T_{2^t}(a)=\sum_{0\le i<t} a^{2^i}.
\]
Further we define the {\em subtrace} of $a$ by
\[
\S_{2^t}(a)=\sum_{0\le i<j<t} a^{2^i+2^j}.
\]
Note $S_{2^t}(a)\in\F_2$ and it is denoted $e_2$ in \cite{mod256}.

We obtain a necessary condition for $a$ to be
a Kloosterman zero using the following characterization
of Kloosterman sums modulo~16.

\begin{proposition}
\label{prop-K-16}
Assume that $m\ge 4$ and $a\in \F_{2^m}^*$.
If $\K_{2^m}(a)$ is divisible by~16, then 
$\T_{2^m}(a)=0$ and $\S_{2^m}(a)=0$.
\end{proposition}
\begin{proof}
This follows from Theorem~1 and Theorem~2 in \cite{mod256},
or from Theorem~3.5 in \cite{KZerosNotInSbflds}.
\end{proof}

We now state our new results.

\begin{theorem}
\label{thm-6k-k}
Assume that $k$ is odd, $k\ge 3$, and denote $m=3k$.
Assume that $a\in \F_{2^{2m}}^*$.
Then $f(x)= \Tr^{2m}_k(a x^{2^m-1})$ is not bent.
\end{theorem}
\begin{proof}
Towards a contradiction assume that $f(x)$ is bent. 
As was mentioned earlier,
we can assume without loss of generality
that $a\in \F_{2^m}^*$.
By Propositions \ref{prop-K-vanishes-on-coset}
and \ref{prop-K-16} it follows that
\begin{equation}
\label{eq-S-condition}
\T_{2^m}(az)=\S_{2^m}(az)=0 \mbox{\ \ \ for\ each\ } z\in\F_{2^k}.
\end{equation}

Consider the bivariate polynomials
\[
C(A,u)=\sum_{0\le i<m} A^{2^i}u^{2^i\bmod (2^k-1)}
\]
and
\[
D(A,u)=\sum_{0\le i<j<m} A^{2^i+2^j}u^{(2^i+2^j)\bmod (2^k-1)}.
\]
We can also view $C$ and $D$
as univariate polynomials in $u$ 
and write them as 
\[
C(A,u)=\sum_{i=0}^{2^k-2}C_i(A)u^i
\]
and
\[
D(A,u)=\sum_{i=0}^{2^k-2}D_i(A)u^i
\]
where $C_i(A)$ and $D_i(A)$ are polynomials in $A$ with coefficients
in $\F_2$.

For integers $s,t$ where $t\ge 0$ and $tk\le s < (t+1)k$ we have
\begin{equation}
\label{eq-red}
2^s\bmod (2^k-1)=2^{s-tk}.
\end{equation}
Indeed $2^s=2^{s-tk}(2^{tk}-1)+2^{s-tk}$
and $2^k-1$ divides $2^{tk}-1$.
Equation (\ref{eq-red}) allows one to write down
closed forms for some of the polynomials $C_i$ and $D_i$. In particular
\begin{eqnarray*}
C_1(A)& =&\sum_{0\le i\le 2} A^{2^{ik}} 
%\label{eq-C1}
\\
D_1(A)& =&\sum_{1\le i<j\le 3} A^{2^{ik-1}+2^{jk-1}}.
%\label{eq-D1}
\end{eqnarray*}

We note that 
$\T_{2^m}(bz)=C(b,z)$
and
$\S_{2^m}(bz)=D(b,z)$
for each $b\in\F_{2^m}$
and each $z\in\F_{2^k}$.
Recall that 
the constant $a\in\F_{2^m}^*$ is the monomial coefficient of the putative
vectorial bent function, see the statement
of the theorem and the beginning of this proof. 
From (\ref{eq-S-condition}) 
we get $C(a,z)=D(a,z)=0$ 
for each $z\in\F_{2^k}$. Hence $C(a,u)$,
which is a univariate polynomial in $u$ of degree at most $2^k-2$
with coefficients in $\F_{2^m}$,
must be of the form $C(a,u)=P(u)\cdot (u^{2^k}-u)$ for some polynomial $P(u)$.
By comparing degrees we conclude that $C(a,u)$ is the zero polynomial,
in other words $C_i(a)=0$ for all $i$,
in particular 
\begin{equation}
\label{C1-0}
C_1(a)=0.
\end{equation}
By a similar argument we conclude
\begin{equation}
\label{D1-0}
D_1(a)=0.
\end{equation}

We rewrite the polynomial $D_1(A)$ 
in terms of a new polynomial $G(A)$ as
\begin{equation}
\label{D1-rew}
D_1(A)=\left( A^{2^k+1}+ A^{2^{2k}+1}+ A^{2^{2k}+2^k}\right)^{2^{k-1}}
= G(A)^{2^{k-1}}. 
\end{equation}
Now we calculate
\begin{eqnarray}
G(A)+A^{2^k}C_1(A)&=&A^{2^{2k}+1}+A^{2^{k+1}}
=A^{2^k+1}\left(A^{2^k-1}+A^{2^{2k}-2^k}\right)
\nonumber
\\\
&=&A^{2^k+1}\left(A^{2^k-1}+\left(A^{2^k-1}\right)^{2^k}\right).
\label{finito}
\end{eqnarray}
Since $a\neq 0$ we conclude from (\ref{C1-0})--(\ref{finito})
that $a^{2^k-1}\in\F_{2^k}^*$.

Let $\al$ be a primitive element in $\F_{2^m}=\F_{2^{3k}}$.
We can write non-zero elements of $\F_{2^k}$
in terms of $\al$ as
\[
\F_{2^k}^* = \{ \al^{i\frac{2^{3k}-1}{2^k-1}} \: :\: 0\le i < 2^k-1 \}.
\]
We note that 
\[
\frac{2^{3k}-1}{2^k-1}=2^{2k}+2^k+1=(2^k+2)(2^k-1)+3.
\]
Let $j$ be such that $a=\al^j$.
Since $a^{2^k-1}\in\F_{2^k}^*$, 
there exists integer $s$ such that
\[
j(2^k-1)\equiv s((2^k+2)(2^k-1)+3) \pmod{2^{3k}-1}.
\]
Since $k$ is assumed to be odd,
$\gcd(3,2^k-1)=1$.
Therefore $2^k-1$ divides $s$
and 
\[
a^{2^k-1}=\left(\al^j\right)^{(2^{3k}-1)\frac{s}{2^k-1}}=1.
\] 
Hence $a\in\F_{2^k}$.
By Proposition~\ref{prop-K-vanishes-on-coset}
we have $\K_{2^m}(a)=0$.
This however is a contradiction to Theorem~\ref{thm:KZerosNotInSbflds}.
\end{proof}

By an easy modification
of the final parts of the previous proof
we see that in the case when $k$ is even
we conclude that $a^{2^k-1}$ is a power 
of the primitive cube root of unity in $\F_{2^k}$.
We record this in the next statement.

\begin{theorem}
\label{thm-6k-k-even}
Assume that $k$ is even and denote $m=3k$.
Without loss of generality
assume that $a\in \F_{2^m}^*$.
If $f(x)= \Tr^{2m}_k(a x^{2^m-1})$ is bent,
then $a^{3(2^k-1)}+1=0$.
\end{theorem}

\begin{example}
Theorem~\ref{thm-6k-k-even} can be used for a quick and simple discovery
of Dillon monomial bent functions from $\F_{2^{12}}$ to $\F_{2^2}$.
Let $k=2$ in Theorem~\ref{thm-6k-k-even},
which leads us to the equation $a^9+1=0$. 
Now let $a$ be any of the primitive 9-th roots of unity in $\F_{2^6}$,
that is, $a$ is a root of $a^6+a^3+1=0$. Then it can be 
easily verified using Proposition~\ref{prop-K-vanishes-on-coset}
that $f(x)= \Tr^{12}_2(a x^{2^6-1})$ is bent.
These bent functions were first found 
by computer search \cite[Theorem~8]{LL-ISIT16}.
\end{example}

\section{Conclusion}

In order to prove new and stronger non-existence results
for vectorial monomial Dillon type bent functions,
we introduced new proof techniques.
We combined the trace condition and the subtrace condition
for Kloosterman zeros, the latter one
being a result of our previous research \cite{mod256}
on divisibility properties of Kloosterman sums
modulo powers of~2.
We constructed a system of equations involving bivariate
polynomials which allowed us to derive strong
conditions on the coefficient of the putative bent monomial.
Since we have not by far exhausted the results of \cite{mod256},
we believe that further proofs may be inspired
by this paper.

\section*{Acknowledgement}

Research of both authors 
was supported
in part by the Natural Sciences and Engineering Research Council of Canada
(NSERC)
under the research grant RGPIN-2015-06250.

\end{document}